\documentclass[10pt,conference]{IEEEtran}
\normalsize
\usepackage{amssymb,amsmath,amsthm,bm}
\usepackage{enumitem}
\usepackage{cite}

\usepackage[ruled]{algorithm}
\usepackage[noend]{algorithmic}
\makeatletter
\newcommand\fs@norules{\def\@fs@cfont{\bfseries}\let\@fs@capt\floatc@ruled
	\def\@fs@pre{}%
	\def\@fs@post{}%
	\def\@fs@mid{\kern3pt}%
	\let\@fs@iftopcapt\iftrue}
\makeatother

\usepackage[utf8]{inputenc}
\usepackage[T1]{fontenc}

\usepackage{graphicx}
\usepackage{xcolor} 
\usepackage{tikz,tikzscale}
\usetikzlibrary{fit, spy, calc, shapes, arrows, positioning, plotmarks, arrows.meta, pgfplots.groupplots, backgrounds}
\usepackage{pgfplots}
\pgfplotsset{compat=newest}
\usepgfplotslibrary{patchplots}

\makeatletter
\let\MYcaption\@makecaption
\makeatother
\usepackage[font=footnotesize]{subcaption}
\makeatletter
\let\@makecaption\MYcaption
\makeatother

\renewcommand{\vec}[1]{{\bm{#1}}}
\DeclareMathOperator*{\argmin}{argmin}
\DeclareMathOperator*{\argmax}{argmax}
\DeclareMathOperator*{\id}{id}
\DeclareMathOperator*{\NC}{NC}
\usepackage{mathtools}
\newcommand{\defeq}{\vcentcolon=}

\usepackage{etoolbox}
\patchcmd{\thmhead}{(#3)}{#3}{}{}

\newtheorem{theorem}{Theorem}
\newtheorem{lemma}{Lemma}

\newtheorem{collorary}{Corollary}
\newtheorem{definition}{Definition}

\newtheorem{example}{Example}

\newcommand*{\SHORT}{}%
\ifdefined\SHORT
\else
\fi

\begin{document}
\bstctlcite{IEEEexample:BSTcontrol}
\ifdefined\SHORT
	\title{\huge Finite-Precision Implementation of Arithmetic Coding Based Distribution Matchers \vspace{-5pt}}
\else
	\title{On Finite-Precision Arithmetic Implementation of Arithmetic Coding Based Distribution Matchers} 
\fi
\author{ 
	Marcin~Pikus\IEEEauthorrefmark{1}\IEEEauthorrefmark{2},
	Wen~Xu\IEEEauthorrefmark{1}, and
    Gerhard~Kramer\IEEEauthorrefmark{2}\\
	\IEEEauthorblockA{\IEEEauthorrefmark{1}Huawei Technologies Duesseldorf GmbH, Munich Research Center, Riesstr. 25, 80992 Munich, Germany\\
		\IEEEauthorblockA{\IEEEauthorrefmark{2}Institute for Communications Engineering, Technical University of Munich, Arcisstr. 21, 80290 Munich, Germany}
		marcin.pikus@tum.de, wen.dr.xu@huawei.com, gerhard.kramer@tum.de \vspace{-1em}
	}

\thanks{M. Pikus is with the Huawei Technologies Duesseldorf GmbH, D-80992 Munich, Germany,
	 and also with the Institute for Communications Engineering, Technische Universität München, D-80333
	Munich, Germany (e-mail: marcin.pikus@gmail.com).}
\thanks{W. Xu is with the Huawei Technologies Duesseldorf GmbH, D-80992 Munich, Germany (e-mail: wen.xu@ieee.org).}
}
\maketitle


\begin{abstract}
A distribution matcher (DM) encodes a binary input data sequence into a sequence of symbols with a desired target probability distribution. Several DMs, including shell mapping and constant-composition distribution matcher (CCDM), have been successfully employed for signal shaping, e.g., in optical-fiber or 5G. The CCDM, like many other DMs, is typically implemented by arithmetic coding (AC). In this work we implement AC based DMs using finite-precision arithmetic (FPA). An analysis of the implementation shows that FPA results in a rate-loss that shrinks exponentially with the number of precision bits.  Moreover, a relationship between the CCDM rate and the number of precision bits is derived.
\end{abstract}

\ifdefined\SHORT
\else
\fi

\IEEEpeerreviewmaketitle
\section{Introduction}
A distribution matcher (DM) reversibly maps a sequence $\vec{U}$ of independent and uniformly distributed bits into a sequence $\vec{A}$ of symbols to emulate a memoryless source $P_A$, i.e., the output of the DM approximates a sequence of independent and identically distributed (IID) symbols, each distributed according to $P_A$.  The accuracy of the approximation is measured by the Kullback–Leibler (KL) divergence between the probability distribution of the DM's output and the probability distribution of the IID sequence.  An inverse distribution matcher (DM$^{-1}$) performs the inverse operation recovering $\vec{U}$  from $\vec{A}$. 

DMs can be used in communication systems, such as probabilistic amplitude shaping (PAS) \cite{Boecherer2015}, to adjust the distribution of transmitted symbols to a distribution beneficial for a certain channel, e.g., a distribution achieving the capacity. PAS with a constant-composition distribution matcher (CCDM) \cite{Schulte2015} was recently used for optical-fiber communication \cite{Buchali2015} and proposed for the 5G mobile system \cite{PSCM_Huawei}. DMs can also be interesting for secrecy in communications.

We refer to CCDM and other DMs implemented by arithmetic coding (AC), e.g., the multi-composition DM \cite{Pikus2019} and the multiset-partition DM \cite{Fehenberger2018}, as AC based distribution matchers (AC-DMs). AC is applied in a reverse order for distribution matching, i.e., the AC-DM is implemented by an AC decompresser and the inverse AC-DM by an AC compresser. A direct implementation of AC requires infinite-precision arithmetic (IPA) operations. A large body of research was dedicated to efficient finite-precision arithmetic (FPA) implementation\footnote{Usually using integer operations.} of AC, see e.g. \cite{Rissanen1976, Pasco1976, Rubin1979, Langdon1984, Witten1987}. An AC-DM for CCDM was first proposed by Ramabadran in \cite{Ramabadran1990} for binary inputs and outputs, and by Schulte \& Böcherer in \cite{Schulte2015} for larger output alphabets.
CCDM has good-performance and low-complexity for long output sequences, whereas other solutions, e.g., \cite{Pikus2019, Gultekin2018a, Schulte2018} are usually too complex\footnote{The memory complexity increases at least linearly with the length of the sequence.} to be used with very long sequences. 

In this work, we show how to adapt the technique from \cite{Ramabadran1990} to efficiently implement non-binary AC-DMs. We show that the FPA implementation decreases the input sequence length, and that the decrease shrinks exponentially with the number of precision bits.  We also derive necessary conditions to ensure that an FPA implementation of a CCDM is one-to-one. The one-to-one property is needed for error-free decoding, but proving it is challenging because of many rounding operations performed by the encoding and decoding algorithms. Moreover, we show that the CCDM is not asymptotically optimal in terms of KL divergence when implemented in FPA. This however does not significantly affect performance in practice. 

The work is organized as follows. Sec. \ref{s:dm_fund} introduces distribution matching. Sec. \ref{s:ac_based_dm} and  Sec. \ref{s:finite_prec_impl} describe how AC-DMs are implemented if one could use IPA, and when one uses FPA. Sec. \ref{s:bounding_the_d}  shows how to guarantee that an AC-DM is one-to-one and Sec. \ref{s:eff_ccdm} gives specific results for the CCDM. Sec. \ref{s:results} applies the ideas to several DMs.

We denote random variables (RVs) by uppercase letters, such as $A$, and realizations by lowercase letters, such as $a$. A row vector is denoted by a bold symbol, e.g., $\vec{a}$. The $i$-th entry in the vector $\vec{a}$ is denoted by $a_i$, and a subvector $[a_i,a_{i+1},\cdots, a_j]$ of $\vec{a}$ is denoted by $\vec{a}_i^j$. The length (dimension) of a vector is denoted by $l(\vec{a})$, e.g., we have $\vec{a} = \vec{a}_1^{l(\vec{a})}$. A RV uniformly distributed on a set $\mathcal{A}$ is denoted by $\mathbb{U}_{\mathcal{A}}$, i.e., $S \sim \mathbb{U}_{\mathcal{A}}$ means that $P_{S}(s)=1/|\mathcal{A}|$ for $s \in \mathcal{A}$.

\section{Distribution Matching} \label{s:dm_fund}
A one-to-one block-to-block DM is an injective function $f_{\text{DM}}$ from binary input sequences $\vec{u} \in \{0,1\}^k$ to codewords $\vec{c}$ from the codebook $\mathcal{C}$:
\begin{equation}
f_{\text{DM}} \colon  \{0,1\}^k \to \mathcal{C} \subseteq  \mathcal{A}^n
\end{equation}
where $\mathcal{A}$ is the output alphabet. The ratio $R = \frac{k}{n}$ is called the \textit{matching rate}. A higher matching rate for a given output distribution results in a higher transmission rate. We assume that the input sequence $\vec{U}$ is a random vector consisting of $k$ IID Bernoulli$(1/2)$ distributed bits. The output sequence of the DM is thus a random vector $\vec{\tilde{A}} = f_{\text{DM}}(\vec{U}) \sim \mathbb{U}_{\mathcal{C}}$ uniformly distributed on $ \mathcal{C}$. The goal of the DM is to make its output "look" as if it was a sequence of IID RVs, each distributed according to a target probability distribution $P_A$. This is usually performed by minimizing the normalized KL divergence between the DM's output $\vec{\tilde{A}}$ and the IID sequence $\vec{A} \sim P_A^n=\prod_{i=1}^n P_A:$
\begin{equation}
\frac{1}{n}\mathbb{D}(P_{\vec{\tilde{A}}} \| P_A^n) = \frac{1}{n} \sum_{\vec{c} \in \mathcal{C}}  \frac{1}{|\mathcal{C}|} \log_2 \frac{\frac{1}{|\mathcal{C}|}}{ P_A^n(\vec{c})}. \label{eq:def_div}
\end{equation} 
Different approaches for implementing low-complexity mappings $f_{\text{DM}}$ that achieve low normalized divergence can be found, e.g., in \cite{Schulte2015, Pikus2017, Pikus2019, Fehenberger2018, Gultekin2018a, Schulte2018, Steiner2018, Iscan2018a}.

\section{Arithmetic Coding Based DM} \label{s:ac_based_dm}
An AC-DM maps binary\footnote{Extensions to larger input alphabet sizes are straightforward.} input sequences $\vec{u}$ of length $k$ to non-binary sequences (codewords) $\vec{c}$ of length $n$. The sequences $\vec{c}$ are formed by symbols from the output alphabet $\mathcal{A} = \{a_1, \dots, a_m\}$, i.e., $\vec{c} \in \mathcal{A}^n$.

Each input sequence $\vec{u}_i, i\!=\!1,\dots,2^k,$ corresponds to a distinct point $d(\vec{u}_i), i\!=\!1,\dots,2^k,$ from the interval $[0,1)$. On the other hand, each codeword $\vec{c} \in \mathcal{A}^n$ corresponds to a distinct subinterval $I(\vec{c})$ (possibly of zero length) of the interval $[0,1)$. The subintervals $I(\vec{c})$ are chosen such that they \textit{partition} $[0,1)$, i.e., they are pairwise disjoint and $\bigcup_{\vec{c} \in \mathcal{A}^n} I(\vec{c}) = [0,1)$. At the encoder an input data sequence $\vec{u}$ is mapped to a codeword $\vec{c}$ if the point $d(\vec{u})$ lines inside the interval $I(\vec{c})$. At the decoder first an interval $I(\vec{c})$ is determined based on  the received codeword $\vec{c}$. Then, a point $d(\vec{u}) \in I(\vec{c})$ is found and decoded to the sequence $\vec{u}$.  
\begin{definition}[Natural $m$-ary code number]
	Consider the alphabet $\mathcal{A}=\{a_1, \dotsc, a_m\}$ and a sequence $\vec{x} \in \mathcal{A}^n$. The function ${\NC}_m(\cdot)$  returns a natural $m$-ary code number corresponding to the sequence $\vec{x}$, i.e., 
	\begin{equation}
	{\NC}_m(\vec{x}) = \sum_{j=1}^n \left(\id(x_j)-1\right) m^{n-j} 
	\end{equation}
	where the function $\id(\cdot)$ returns the alphabet index of the symbol, i.e., $\id(a_i) = i, \;\forall i$. 
\end{definition}

A binary input sequence $\vec{u} \in \{0,1\}^k$ is mapped to a point $d(\vec{u}) \in [0,1)$ via
\begin{equation}
d(\vec{u}) = \frac{{\NC}_2(\vec{u})}{2^k}. \label{eq:du}
\end{equation}

The codeword's intervals are ordered lexicographically in the $[0,1)$ interval, with the first codeword's symbol $c_1$ being the most significant symbol. We consider the lexicographical ordering of the output alphabet symbols $a_1<a_2<\dotsc<a_m$.  That is, for two codewords $\vec{c}_1 \in \mathcal{A}^n$ and $\vec{c}_2 \in \mathcal{A}^n$, if ${\NC}_m(\vec{c}_1)<{\NC}_m(\vec{c}_2)$, then $I(\vec{c}_1)$ will be placed somewhere below $I(\vec{c}_2)$. We describe $I(\vec{c})$ by the beginning $x(\vec{c})$ and the width $y(\vec{c})$, i.e., $I(\vec{c}) = [x(\vec{c}), x(\vec{c})+y(\vec{c})).$
An interval $I(\vec{c})$ can be computed recursively using a chosen probability model $P_{\vec{C}}$ on the codeword's symbols. $P_{\vec{C}}$ is usually specified in terms of the conditional probabilities (also referred to as branching probabilities) of the next symbol given the previous symbols, i.e., $P_{C_{i+1}|\vec{C}_1^i}(\cdot|\vec{s})$, where $\vec{s}$ is a sequence denoting a prefix of the codeword. The conditional cumulative probability of a letter $c \in \mathcal{A}$ is defined as
\begin{equation}
F_{C_{i+1}|\vec{C}_1^i}(c|\vec{s}) = \sum_{a \le c} P_{C_{i+1}|\vec{C}_1^i}(a|\vec{s}) \label{eq:cumm_model}
\end{equation}
where $a \le c$ refers to the lexicographical ordering of the alphabet's symbols. Clearly, we have $F_{C_{i+1}|\vec{C}_1^i}(a_m|\vec{s}) = 1$ for any $\vec{s}$. For notational convenience we also use  $F_{C_{i+1}|\vec{C}_1^i}(a_0|\vec{s}) = 0$ if $a_0 \notin \mathcal{A}$.  The computation of the codewords' intervals  can be performed by iteratively applying equations (\ref{eq:x_update}) and (\ref{eq:y_update}) below for $i=0,\dots,n-1:$
\begin{align}
&x(\lambda) =0, \; y(\lambda) =1 \label{eq:xy_init}\\
&x(\vec{s}a_j)\!=\!x(\vec{s})\!+\!y(\vec{s})F_{C_{i+1}|\vec{C}_1^i}(a_{j-1}|\vec{s}) \text{, } j=1,\dotsc,m \label{eq:x_update}\\
&y(\vec{s}a_j) = y(\vec{s})P_{C_{i+1}|\vec{C}_1^i}(a_j|\vec{s})  \text{, } j=1,\dotsc,m \label{eq:y_update}
\end{align}
where  $\lambda$ denotes an empty sequence, and $\vec{s}a_j$ denotes a concatenation of $\vec{s}$ and $a_j$. The recursive procedure (\ref{eq:xy_init})--(\ref{eq:y_update}) gives
\begin{align}
&x(\vec{c}) = \sum_{\vec{c}^\prime \in \mathcal{A}^n\colon \vec{c}^\prime < \vec{c} }  P_{\vec{C}}(\vec{c}^\prime)\\
&y(\vec{c}) = \prod_{i=0}^{n-1} P_{C_{i+1}|\vec{C}_1^i}(c_{i+1} | \vec{c}_1^i) =  P_{\vec{C}}(\vec{c})  \label{eq:inf_prec_y}
\end{align}
where $\vec{c}^\prime < \vec{c}$ refers to the lexicographical ordering of the codewords.

A one-to-one mapping between data sequences and codewords can be established if each point $d(\vec{u}_i),i=1,\dotsc,2^k,$ belongs to an interval and if each interval  $I(\vec{c}), \vec{c} \in \mathcal{A}^n,$ contains at most one point $d(\vec{u})$. The first condition follows because the intervals $I(\vec{c}), \vec{c} \in \mathcal{A}^n,$ partition the unit interval. The second condition can be guaranteed by letting the distance between two adjacent points be greater than the largest interval, i.e.,
\begin{equation}
\frac{1}{2^k} \ge \max_{\vec{c} \in \mathcal{A}^n} |I(\vec{c})| = \max_{\vec{c} \in \mathcal{A}^n} |y(\vec{c})| = \max_{\vec{c} \in \mathcal{A}^n} |P_{\vec{C}}(\vec{c})|. \label{eq:interval_ineq}
\end{equation}
We are interested in maximizing $k$ and thus we often choose 
\begin{equation}
k = k_{\text{IPA}} = \left\lfloor -\log_2\left(\max_{\vec{c} \in \mathcal{A}^n} |P_{\vec{C}}(\vec{c})|\right) \right\rfloor. \label{eq:inf_prec_k}
\end{equation}

\section{Finite-Precision Arithmetic Implementation} \label{s:finite_prec_impl}
The above described procedure requires IPA operations in general, which is infeasible in practice. Ramabadran in \cite{Ramabadran1990} describes how to implement a binary CCDM using FPA. We adapt this technique to implement a non-binary AC-DM with an arbitrary model $P_{\vec{C}}$. Instead of using the IPA models $P_{C_{i+1}|\vec{C}_1^i}$, $F_{C_{i+1}|\vec{C}_1^i}$, we use a finite-precision integer\footnote {Bounded integers can be represented using a finite number of bits.} representation $\hat{F}_{\vec{C}}$ for the cumulative model $F_{\vec{C}}$ from (\ref{eq:cumm_model}):
\begin{equation}
\hat{F}_{C_{i+1}|\vec{C}_1^i}(c|\vec{s}) = \left\lfloor \Theta F_{C_{i+1}|\vec{C}_1^i}(c|\vec{s}) + \frac{1}{2}  \right\rfloor \label{eq:fp_model_F}
\end{equation}
with $\hat{F}_{C_{i+1}|\vec{C}_1^i}(a_0|\vec{s})=0$, since $a_0 \notin \mathcal{A}$.  The model $\hat{P}_{\vec{C}}$ is defined as
\begin{equation}
\hat{P}_{C_{i+1}|\vec{C}_1^i}(a_j|\vec{s}) = \hat{F}_{C_{i+1}|\vec{C}_1^i}(a_j|\vec{s}) - \hat{F}_{C_{i+1}|\vec{C}_1^i}(a_{j-1}|\vec{s}). \label{eq:fp_model_P} 
\end{equation}
$\Theta$ is a scaling factor used for the integer representation. It effectively converts the probability models $P_{\vec{C}}, F_{\vec{C}}$ into frequency-counts models $\hat{P}_{\vec{C}}, \hat{F}_{\vec{C}}$ per $\Theta$ symbols. Next, we represent the subsequent intervals appearing in (\ref{eq:xy_init})--(\ref{eq:y_update}) by three integer numbers $\hat{x}(\vec{s}), \hat{y}(\vec{s})$, and $L(\vec{s})$. The start $x(\vec{s})$ and width $y(\vec{s})$ of the interval will be represented as binary fractions with $L(\vec{s})+w$ bits ($w$ is a fixed parameter that we choose as described in (\ref{eq:2w_theta}) below)
\begin{equation}
x(\vec{s}) = \frac{\hat{x}(\vec{s})}{2^{L(\vec{s})+w}},\;\;\;\; y(\vec{s}) = \frac{\hat{y}(\vec{s})}{2^{L(\vec{s})+w}}. \label{eq:fp_y_yhat}
\end{equation}
The recursive formulas for computing the values  $\hat{x}(\vec{s}), \hat{y}(\vec{s})$, and $L(\vec{s})$ are
\begin{align}
&\hat{x}(\lambda)=0, \;\hat{y}(\lambda)=2^w, \;L(\lambda)=0 \label{eq:fp_xy_init}\\
&\hat{x}(\vec{s}a_j) \!=\! \left(\!\hat{x}(\vec{s})\!+\!\left\lfloor\!\frac{\hat{y}(\vec{s}) \hat{F}_{C_{i+1}|\vec{C}_1^i}(a_{j-1}|\vec{s})}{\Theta}\!+\!\frac{1}{2}\!\right\rfloor \right)2^v \label{eq:fp_x_update} \\
&\hat{y}(\vec{s}a_j) = \left( \left\lfloor \frac{\hat{y}(\vec{s}) \hat{F}_{C_{i+1}|\vec{C}_1^i}(a_j|\vec{s})}{\Theta} + \frac{1}{2} \right\rfloor \right. + \nonumber \\
&\hspace{1.8cm}- \left.\left\lfloor \frac{\hat{y}(\vec{s}) \hat{F}_{C_{i+1}|\vec{C}_1^i}(a_{j-1}|\vec{s})}{\Theta} + \frac{1}{2} \right\rfloor \right)2^v \label{eq:fp_y_update} \\
&L(\vec{s}a_j) = L(\vec{s}) + v \label{eq:fp_L_update}
\end{align}
where $v$ is chosen such that\footnote{If $\hat{P}_{C_{i+1}|\vec{C}_1^i}(a_j|\vec{s}) = 0$ for certain $\vec{s},a_j$, then we may have $\hat{y}(\vec{s}a_j)=0$ and a $v$ satisfying (\ref{eq:fp_v_update}) can not be found. This does not lead to problems as the encoder will not produce codewords with the prefix $\vec{s}a_j$ (the codeword's interval has zero length), so further subdivisions are not needed.}
\begin{equation}
2^{w} \le \hat{y}(\vec{s}a_j) < 2^{w+1}. \label{eq:fp_v_update}
\end{equation}
The parameter $w+1$ represents the number of bits used to represent the mantissa $\hat{y}(\vec{s}a_j)$ of the current interval width $y(\vec{s}a_j)$. The scaling by $2^v$ in (\ref{eq:fp_x_update}) and (\ref{eq:fp_y_update}) guarantees that the mantissa $\hat{y}(\vec{s}a_j)$ is at least $2^{w}$. This provides sufficient precision for further subdivisions in (\ref{eq:fp_x_update}) and (\ref{eq:fp_y_update}) for the next symbols. Note that (\ref{eq:fp_xy_init})--(\ref{eq:fp_y_update}) round the interval boundaries rather than the width. The interval width (\ref{eq:fp_y_update}) is simply a difference between the two boundaries. This ensures that the original interval is partitioned during each step. Thus, the codeword intervals $I(\vec{c}), \vec{c} \in \mathcal{A}^n$, partition the starting unit interval. We want to avoid having the intervals disappear due to the rounding operations, i.e., we require
\begin{equation*}
\hat{P}_{C_{i+1}|\vec{C}_1^i}(a_j|\vec{s}) > 0 \implies \hat{y}(\vec{s}a_j)>0
\end{equation*} 
which will be the case if $\frac{\hat{y}(\vec{s}a_j)}{\Theta} \ge 1$. This in turn can be guaranteed by choosing 
\begin{equation}
2^w \ge \Theta. \label{eq:2w_theta}
\end{equation}
As in the IPA case (\ref{eq:interval_ineq}), we guarantee error-free decoding by choosing
\begin{equation}
y(\vec{c}) = \frac{\hat{y}(\vec{c})}{2^{L(\vec{c})+w}} \le \frac{1}{2^k}, \;\forall \vec{c} \in \mathcal{A}^n. \label{eq:fp_interval_ineq}
\end{equation}

\section{Bounding the Discrepancy} \label{s:bounding_the_d}
The above described FPA scheme implements a one-to-one mapping if (\ref{eq:fp_interval_ineq}) is satisfied and the codeword intervals partition the unit interval. The latter condition is inherently satisfied by rounding the interval boundaries in (\ref{eq:fp_x_update})--(\ref{eq:fp_y_update}) rather than rounding the interval length. The FPA condition (\ref{eq:fp_interval_ineq}) does not follow from the IPA condition (\ref{eq:interval_ineq}). This is because the intervals computed by the FPA implementation are approximations of the IPA intervals. The discrepancy is due to the model rounding (\ref{eq:fp_model_F})--(\ref{eq:fp_model_P}) and the rounding operations during the computation of (\ref{eq:fp_xy_init})--(\ref{eq:fp_y_update}). To assess the FPA condition (\ref{eq:fp_interval_ineq}) we must bound the rounding error. From (\ref{eq:fp_y_update}) we have
\begin{align}
\frac{\hat{y}(\vec{s}a_j)}{2^v} &=  \left\lfloor \frac{\hat{y}(\vec{s}) \hat{F}_{C_{i+1}|\vec{C}_1^i}(a_j|\vec{s})}{\Theta} + \frac{1}{2} \right\rfloor \nonumber\\
&\hspace{1.5cm}- \left\lfloor \frac{\hat{y}(\vec{s}) \hat{F}_{C_{i+1}|\vec{C}_1^i}(a_{j-1}|\vec{s})}{\Theta} + \frac{1}{2} \right\rfloor \\
&\le \frac{\hat{y}(\vec{s}) \hat{P}_{C_{i+1}|\vec{C}_1^i}(a_{j}|\vec{s}) }{\Theta} + 1 \label{eq:fp_y_bound}
\end{align}
where the second line follows by $x-1 < \lfloor x \rfloor \le x$ and (\ref{eq:fp_model_P}). Dividing both sides by $2^{L(\vec{s})+w}$ and using $(\ref{eq:fp_y_yhat})$ we get
\begin{align}
y(\vec{s}a_j) &\le y(\vec{s}) \frac{\hat{P}_{C_{i+1}|\vec{C}_1^i}(a_{j}|\vec{s})}{\Theta} + 2^{-L(\vec{s})-w} \\
&\le y(\vec{s}) \left( P_{C_{i+1}|\vec{C}_1^i}(a_{j}|\vec{s}) + \epsilon \right)+ 2^{-L(\vec{s})-w} \\
&\le  y(\vec{s}) P_{C_{i+1}|\vec{C}_1^i}(a_{j}|\vec{s})   \left(1 + \frac{\epsilon + 2^{-w}}{P_{C_{i+1}|\vec{C}_1^i}(a_{j}|\vec{s})} \right) \label{eq:fp_single_interval_mismatch}
\end{align}
where in the second line we introduced $\epsilon$ to denote the maximal absolute error between the IPA model $P_{C_{i+1}|\vec{C}_1^i}$ and the rounded model $\frac{1}{\Theta}\hat{P}_{C_{i+1}|\vec{C}_1^i}$, e.g., $\epsilon = \frac{1}{2} \Theta^{-1}$ if rounding is used. The third line follows because $y(\vec{s}) \ge 2^{-L(\vec{s})}$. Finally, the length of the interval can be bounded by applying (\ref{eq:fp_single_interval_mismatch})  for all codeword symbols consecutively
\begin{equation}
y(\vec{c}) \le P_{\vec{C}}(\vec{c}) \prod_{i=0}^{n-1} \left( 1 + \frac{\epsilon + 2^{-w}}{P_{C_{i+1}|\vec{C}_1^i}(c_{i+1}|\vec{c}_1^i)} \right) \label{eq:fp_final_interval_mismatch}
\end{equation}
which results in the bound (\ref{eq:fp_interval_ineq}) on the input length becoming 
\begin{equation}
k_{\text{FPA}} \!\le\! -\log_2 \!\left( \max_{c \in \mathcal{A}^n} P_{\vec{C}}(\vec{c}) \prod_{i=0}^{n-1} \left(\! 1 + \frac{\epsilon + 2^{-w}}{P_{C_{i+1}|\vec{C}_1^i}(c_{i+1}|\vec{c}_1^i)} \right)\!\right).  \label{eq:fp_interval_ineq2}
\end{equation}
Inequality (\ref{eq:fp_interval_ineq2}) connects the length $k_{\text{FPA}}$ obtained for the FPA implementation and the length $k_{\text{IPA}}$ for the IPA implementation from (\ref{eq:inf_prec_k}). Comparing (\ref{eq:fp_final_interval_mismatch}) and (\ref{eq:inf_prec_y}) we observe that the base intervals can dilate due to the model approximation and rounding operations. This will effectively reduce the input  length. Using higher precision arithmetic, i.e., larger $w$ and $\Theta$, results in a lower dilatation. 

We emphasize the importance of (\ref{eq:fp_interval_ineq2}), since checking directly\footnote{For example by encoding and decoding all possible input sequences.} if an AC-DM is one-to-one is not feasible for long sequences. An alternative approach may involve evaluating the right-hand-side of (\ref{eq:fp_interval_ineq2}) for some randomly selected codewords from the codebook of the AC-DM and selecting the lowest obtained $k_{\text{FPA}}$. This approach guarantees error free decoding with high probability. Interestingly, $k_{\text{FPA}}$ for the CCDM has a closed form expression. See next section.

We can get a more-restrictive upper bound on $k_{\text{FPA}}$ by decomposing the maximization, i.e., 
\begin{align*}
k_{\text{FPA}} &\le -\log_2 \!\left( \max_{\vec{c} \in {\text{supp}(P_{\vec{C}})}} P_{\vec{C}}(\vec{c}) \right) + \\
&- \underbrace{\max_{\vec{c} \in {\text{supp}(P_{\vec{C}})}} \;\; \sum_{i=0}^{n-1} \log_2\left(  1 + \frac{\epsilon + 2^{-w}}{P_{C_{i+1}|\vec{C}_1^i}(c_{i+1}|\vec{c}_1^i)} \right)}_{\Delta k}
\end{align*}
where $\text{supp}(P_{\vec{C}}) = \{\vec{c} \in \mathcal{A}^n \colon P_{\vec{C}}(\vec{c})>0 \}$. 
The first term is the upper bound on $k_{\text{IPA}}$ from (\ref{eq:inf_prec_k}). The latter term is the input length loss (also called rate-loss) $\Delta k$  due to the FPA implementation of an AC-DM ($\Delta k \!=\!0 \implies k_{\text{FPA}}\!=\!k_{\text{IPA}}$). By using the identity $\log_2(1+x) \le x \log_2 e$, we get 
$$\Delta k \le (\epsilon + 2^{-w})  \left( \max_{\vec{c} \in {\text{supp}(P_{\vec{C}})}} \sum_{i=0}^{n-1} \frac{\log_2 e}{P_{C_{i+1}|\vec{C}_1^i}(c_{i+1}|\vec{c}_1^i)} \right)$$ 
The rate-loss $\Delta k$ shrinks exponentially\footnote{If $\epsilon$ shrinks exponentially which is the case when rounding is used.} with $w$, which allows to keep the rate loss small with reasonable precision.

\section{Efficient Implementation of the CCDM} \label{s:eff_ccdm}
We now turn to designing an FPA CCDM.
\begin{definition}[Composition and $n$-type]
	A \textit{composition} of a vector $\vec{c} \in \mathcal{A}^n$ is a vector containing the numbers of occurrences in $\vec{c}$ of each of the symbols from the alphabet $\mathcal{A}$. We denote a composition by 
	\begin{equation}
	\gamma(\vec{c}) \defeq [n_{a_1}(\vec{c}), \dotsc,  n_{a_m}(\vec{c})]
	\end{equation}
	where $n_{a}(\vec{c}) = |\{i \colon c_i = a\}|$ denotes the number of occurrences of the symbol $a$ in the sequence $\vec{c}$. An $n$-type $Q_A$ is a probability distribution corresponding to the composition $\gamma(\vec{c})$
	\begin{equation}
	Q_A(a) = \frac{n_{a}(\vec{c})}{n}, \;\; a \in \mathcal{A}.
	\end{equation}		
\end{definition}
\begin{example} 
$\mathcal{A} \!=\! \{0, 1\}, \vec{c} = [1011], \vec{\gamma}(\vec{c}) = [1,3]$ and $Q_A(0)=0.25, Q_A(0)=0.75.$
\end{example}

The CCDM chooses some composition $\vec{\gamma}=[n_{a_1}, \dotsc, n_{a_m}]$ and the following model for AC
\begin{equation}
P_{\vec{C}}(\vec{c}) =  \begin{cases}
\frac{1}{ | \mathcal{T}_{\vec{\gamma}} | },& \text{if } \vec{c} \in \mathcal{T}_{\vec{\gamma}}\\ \label{eq:pc_model_ccdm}
0,              & \text{otherwise}
\end{cases}							
\end{equation}
where $\mathcal{T}_{\vec{\gamma}}$ is a set of length-$n$ sequences with the composition $\vec{\gamma}$, i.e., $\mathcal{T}_{\vec{\gamma}} = \{\vec{c} \in \mathcal{A}^n \colon \gamma(\vec{c}) = \vec{\gamma} \}$. For the IPA implementation the input length would be $k_{\text{IPA}} = \lfloor \log_2 |\mathcal{T}_{\vec{\gamma}}| \rfloor$. The CCDM's conditional model is
\begin{equation}
P_{C_{i+1}|\vec{C}_1^i}(a_j|\vec{s}) = \frac{n_{a_j} - n_{a_j}(\vec{s})}{n-i}, \text{ for } a_j \in \mathcal{A}. \label{eq:ccdm_p_s}
\end{equation}

The CCDM is an AC-DM with the model (\ref{eq:pc_model_ccdm}) and can be implemented using the FPA implementation presented above. The CCDM's conditional probabilities (\ref{eq:ccdm_p_s}) can be represented exactly\footnote{From (\ref{eq:fp_model_P}), (\ref{eq:fp_model_F}), and  (\ref{eq:cumm_model}), the rounded model $\frac{1}{\Theta}\hat{P}_{C_{i+1}|\vec{C}_1^i}$ with $\Theta=n-i$ is equal to the IPA model $P_{C_{i+1}|\vec{C}_1^i}$.} by choosing  $\Theta=n-i$ in (\ref{eq:fp_model_F}), i.e., $\Theta$ is decremented after each encoded symbol. This way the rounding is avoided and $k_{\text{FPA}}$ from (\ref{eq:fp_interval_ineq2}) can be bounded by
\begin{align}
&k_{\text{FPA}} \!\le\!  \log_2 |\mathcal{T}_{\vec{\gamma}}| - \max_{\vec{c} \in {\mathcal{T}_{\vec{\gamma}}}} \;\; \sum_{i=0}^{n-1} \log_2\left(  1 + \frac{2^{-w}}{P_{C_{i+1}|\vec{C}_1^i}(c_{i+1}|\vec{c}_1^i)} \right)\nonumber\\ 
&= \log_2 |\mathcal{T}_{\vec{\gamma}}|  - \Delta k \stackrel{\text{choose}}{\implies} k_{\text{FPA}} = \lfloor \log_2 |\mathcal{T}_{\vec{\gamma}}| - \Delta k \rfloor \label{eq:ccdm_k_ub} 
\end{align}
where $\vec{\gamma}$ is the composition used by the CCDM. $\Delta k$ can be found analytically for the CCDM, see Theorem \ref{t:ccdm}. This theorem simplifies the choice of $k_{\text{FPA}}$ for an FPA CCDM and guarantees that the CCDM is one-to-one. We note that Theorem \ref{t:ccdm} applies to any composition, i.e., the alphabet's symbols can be relabeled such that the theorem's prerequisites are satisfied. 

\begin{theorem}[The longest interval for FPA CCDM]\label{t:ccdm}
	Consider an FPA CCDM using the composition $\vec{\gamma} = [n_{a_1}, \dotsc, n_{a_m}]$ with $n_{a_1} \le n_{a_2} \le \dotsc \le n_{a_m}$. A sequence 
	\begin{equation}
	\vec{z} = [\underbrace{a_1 \dotsc a_1}_{n_{a_1}} \underbrace{a_2 \dotsc a_2}_{n_{a_2}} \dotsc \underbrace{a_m \dotsc a_m}_{n_{a_m}}], \label{eq:ccdm_worst_case_seq}
	\end{equation} 
	has the largest upper bound (\ref{eq:fp_final_interval_mismatch}) on the interval length among all sequences in $\mathcal{T}_{\vec{\gamma}}$, or equivalently the interval $I(\vec{z})$ can be the longest after dilution due to the rounding operations. Thus, the sequence $\vec{z}$ determines the FPA rate-loss
	\begin{equation}
	\Delta k =  \sum_{i=0}^{n-1} \log_2\left(  1 + \frac{2^{-w}}{P_{C_{i+1}|\vec{C}_1^i}(z_{i+1}|\vec{z}_1^i)} \right). \label{eq:thm_rateloss}
	\end{equation} 	
\end{theorem}
\begin{proof}
	See Appendix.
\end{proof}

In \cite{Schulte2015} the authors consider a one-to-one IPA CCDM that is asymptotically optimal.\footnote{In short, the optimality means that for a target distribution $P_A$ and properly chosen composition, we have $\frac{k}{n} \to \mathbb{H}(P_A)$ and $\frac{1}{n} \mathbb{D}(P_{\vec{\tilde{A}}} \| P_A^n) \to 0$ as $n \to \infty$, see \cite{Schulte2015} for more detail.}
It turns out that the FPA rate-loss $\Delta k$ prevents the asymptotic optimality of a one-to-one FPA CCDM implemented as above. See Corollary 1.

\begin{collorary}[One-to-one, FPA CCDM is not asymptotically optimal] \label{c:ccdm_is_not_optimal}
	Consider an FPA CCDM with the precision parameter $w$.
	\begin{enumerate}
		\item Suppose CCDM uses the $n$-type $Q_A$. The matching rate $R = \frac{k_{\text{FPA}}}{n}$ of the CCDM satisfies 
		\begin{equation}
		\lim_{n \to \infty} R < \mathbb{H}(Q_A) -\log_2 \left(1 + 2^{-w}\right). \label{eq:coll_bound1}
		\end{equation}
		\item Consider an arbitrary target distribution $P_A$ on the output alphabet $\mathcal{A}$, and a CCDM that chooses an arbitrary $n$-type $Q_A$. Then we have
		\begin{equation}
		\lim_{n \to \infty} \frac{1}{n} \mathbb{D}(P_{\vec{\tilde{A}}} \| P_A^n) > \log_2 \left(1 + 2^{-w}\right). \label{eq:coll_bound2}
		\end{equation}	
	\end{enumerate}
\end{collorary}
\begin{proof}
	From (\ref{eq:ccdm_k_ub}) we have
	\begin{align}
	R = \frac{k_{\text{FPA}}}{n} &\le  \frac{1}{n} \log_2 |\mathcal{T}_{\vec{\gamma}}|  - \frac{\Delta k}{n}\\
	&< \mathbb{H}(Q_A) - \log_2 \left(1 + 2^{-w}\right) \label{eq:tmp1}
	\end{align}
	where $\vec{\gamma}$ is the corresponding composition. The inequality  $\frac{1}{n} \log_2 |\mathcal{T}_{\vec{\gamma}}|<\mathbb{H}(Q_A)$ follows by the optimality of the IPA CCDM. Equation (\ref{eq:tmp1}) follows by (\ref{eq:ccdm_k_ub}) with the branching probabilities bounded by one.
	
	Next, the divergence for the CCDM with the $n$-type $Q_A$ is
	\begin{align*}
	\frac{1}{n}\mathbb{D}(P_{\vec{\tilde{A}}} \| P_A^n) &= \mathbb{H}(Q_A) - \frac{k_{\text{FPA}}}{n} + \mathbb{D}(Q_A||P_A)\\
	&> \log_2 \left(1 + 2^{-w}\right)\!+\!\mathbb{D}(Q_A||P_A) 
	\end{align*}
	where the last step follows by (\ref{eq:tmp1}). 
\end{proof}
We remark that the bounds (\ref{eq:coll_bound1}) and (\ref{eq:coll_bound2}) are not tight and suffice only to show that the FPA CCDM is not asymptotically optimal. Tighter bounds can be obtained by evaluating $\Delta k$ using Theorem \ref{t:ccdm}. 

\section{Results} \label{s:results}
\begin{figure}
	\includegraphics{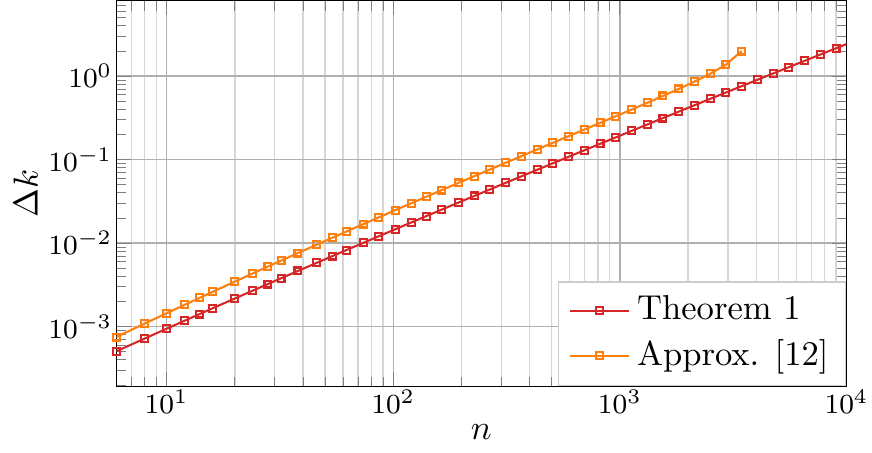}
	\caption{$\Delta k$ for binary CCDM with composition $\vec{\gamma} = [\frac{n}{2}, \frac{n}{2}]$ and $w=14$ obtained by Theorem \ref{t:ccdm} and the bound from \cite{Ramabadran1990}.}
	\label{f:ccdm_ex1}
\end{figure}
\begin{figure*}
	\centering
	\begin{subfigure}[t]{0.47\textwidth}%
		\centering%
		\includegraphics{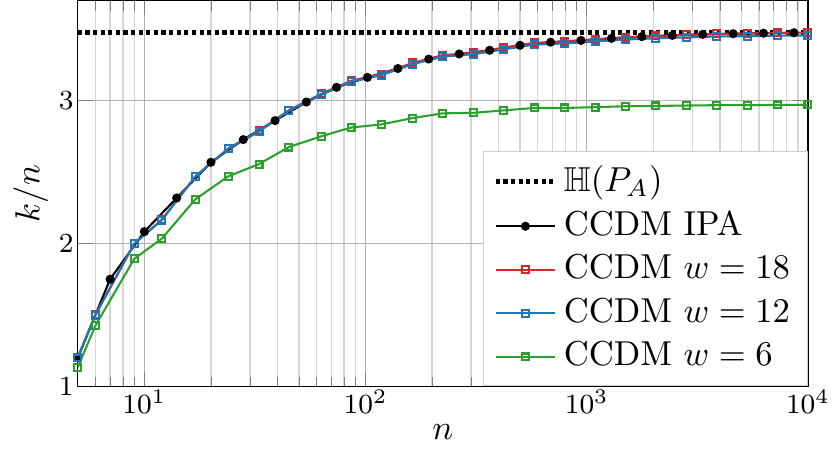}%
		\caption{Matching rate.}%
		\label{f:ccdm_1a}					
	\end{subfigure}%
	~ 
	\begin{subfigure}[t]{0.48\textwidth}
		\centering
		\includegraphics{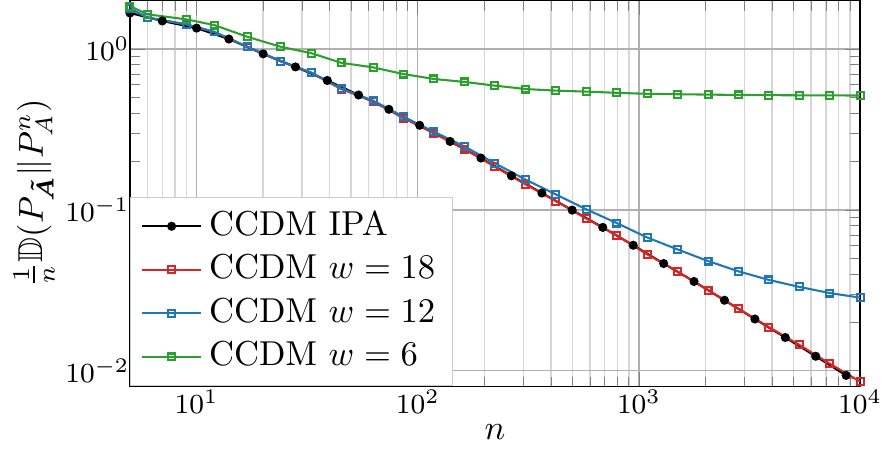}
		\caption{Normalized divergence.}
		\label{f:ccdm_1b}
	\end{subfigure}
	\caption{Matching rate and normalized divergence for CCDM with the target distribution $P_A(a) \propto e^{-0.004 a^2}$ for $a \in \mathcal{A} = \{1, 3, \dotsc, 31\}$, and precision parameter $w \in \{6, 12, 18\}$.}
	\label{f:ccdm_1}	
\end{figure*}

\subsection{Optimal $m$-out-of-$n$ Codes}
The paper \cite{Ramabadran1990} provides an upper bound on the rate-loss $\Delta k$ for a binary CCDM. For the composition $\vec{\gamma} = [n_0, n_1]$, the bound from \cite{Ramabadran1990} is
$$ \Delta k < \log_2\left( 1 + \log_e \frac{1}{1 - 2^{-(w+1)} T }\right) \text{  for  } 2^{-(w+1)} T<1$$
with
$$ T = \sum_{b \in \{0,1\}} n_b \left(1.5772 + \log_e n_{1-b} + \frac{1}{2n_{1-b}} \right).$$

The optimal codes $k = \lfloor \log_2 \binom{n}{m} \rfloor$, can be constructed by the FPA CCDM only when $\Delta k < 1$. Let $n_{\text{MAX}}$ denote the maximum output sequence length for which $\Delta k < 1$. For example, \cite{Ramabadran1990} approximates $n_{\text{MAX}} \approx 2390$ for $m=\frac{n}{2}$ and $w=14$. By using Theorem \ref{t:ccdm} we obtain a tighter bound, i.e., $n_{\text{MAX}} \approx 4440$, see Fig. \ref{f:ccdm_ex1}. To verify the result, we built a CCDM with $\vec{\gamma} = [1600, 1600]$ and $ k=3193$, and successfully verified the encoding and decoding for $10^{10}$ different input sequences. 

\subsection{Low Precision CCDM}
We use the algorithm from Sec. \ref{s:finite_prec_impl} to implement non-binary CCDMs with precision parameters $w \in \{6, 12,18\}$. The target distribution is $P_A(a) \propto e^{-0.004 a^2}$ for $a \in \mathcal{A} = \{1, 3, \dotsc, 31\}$, and the CCDM's composition is chosen to minimize per-symbol divergence as in \cite{Schulte2015}. To guarantee that the CCDMs are one-to-one, we use (\ref{eq:ccdm_k_ub}) and Theorem 1 to choose the maximal possible input length $k_{\text{FPA}}$. The results are presented in Fig. \ref{f:ccdm_1}. Due to the low arithmetic precision, the input length of the CCDM with $w\!=\!6$ does not approach the target entropy $\mathbb{H}(P_A)$. Consequently, the normalized divergence for the low-precision CCDM is bounded away from zero and the CCDM is not asymptotically optimal.

It is difficult to see the difference in matching rates of the CCDMs with  $w\!=\!12$ and $w\!=\!18$ due to the scale, see Fig. 2a. However, for long sequences, the divergence of the CCDM with $w\!=\!12$ differs from the divergence of the CCDM with $w\!=\!18$. This happens due to a small (unobservable in Fig. 2a) difference in matching rates. Finally, the CCDM with $w\!=\!18$ performs as well as an IPA CCDM in the range $n \le 10^4$.

Theorem 1 is useful for choosing the lowest-complexity (minimal precision) CCDM for a given target probability and output length. E.g., for $n\!=\!100$, the CCDMs with $w\!=\!12$ and $w\!=\!18$ perform as well as the IPA CCDM. For $n\!>\!1000$, the CCDM with $w\!=\!18$ has significantly lower divergence. We observed the trend that longer output sequences and larger output alphabets require higher precision $w$ to achieve a performance on par with the IPA CCDM.

\section{Conclusions}
We showed how to efficiently implement an AC-DM in FPA, and how to choose the input length for the CCDM to guarantee error-free decoding. The required input length depends on the precision of the arithmetic operations performed by the CCDM implementation. We observe that the precision of $18$ bits allows to achieve a performance on par with the IPA implementation for an alphabet of size $16$ and output sequences of length up to $10^4$ symbols.

\begin{appendices}
\section{Proof of Theorem 1}
We begin with some lemmas.

\begin{lemma}[Properties of $\log_2 (1+\delta x)$]
	Consider the function $f(x) = \log_2 \left( 1 + \delta x \right)$ for $\delta>0, x\ge0$. Then $f$ has the following properties:
		\begin{align}
		x_2 > x_1 \ge 0 &\implies f(x_2)-f(x_1) \le f(x_2-x_1), \label{eq:ap_f_property1} \\
		x_2 \ge 0,  x_1 \ge 0 &\implies f(x_1 + x_2) \le f(x_1) + f(x_2), \\
		x_1, \dotsc, x_k \ge 0 &\implies f\left(\sum_{i=1}^k x_i\right) \le \sum_{i=1}^k f(x_i). \label{eq:ap_f_property3}
		\end{align}
\end{lemma}

\begin{lemma}[Binary maximizer of the cost function] \label{l:a_binary_miximizer}
	Consider real numbers $x_1 > x_2 > \dotsc > x_k \ge 0$ and the function $f(x) = \log_2 \left( 1 + \delta x \right)$ with $\frac{1}{x_1} \ge \delta >0$. Consider a sequence $\vec{s} \in \{0,1\}^k$ with composition $\gamma(\vec{s})=[n_0,n_1]$ where $n_0 \le n_1$. Define the cost function 
	\begin{equation}
	c(\vec{s}) = \sum_{i=1}^k f\left( \frac{x_i}{n_{s_i} - n_{s_i}(\vec{s}_1^{i-1})}\right). \label{eq:a_cost1}
	\end{equation}	
	Then the maximizer of the cost function is
	\begin{equation}
	\vec{z} = [\underbrace{0 \dots 0}_{n_0} \underbrace{1 \dotsc 1}_{n_1} ] =\argmax_{\vec{s} \in \{0,1\}^k \colon \gamma(\vec{s}) = [n_0,n_1]} c(\vec{s}). \label{eq:a_maximizer}
	\end{equation}
	Furthermore, if $n_0 < n_1$ then the maximizer $\vec{z}$ is unique.
\end{lemma}
\ifdefined\SHORT
\begin{proof}
	The proof was removed due to the $6$-page limit for submissions.
\end{proof}
\else
\begin{proof}
	The proof of (\ref{eq:a_maximizer}) follows by induction. It is easy to check that for $k=2$ the maximizer admits the form (\ref{eq:a_maximizer}). Thus, we assume that (\ref{eq:a_maximizer}) holds and consider the case of $\vec{s} \in \{0,1\}^{k+1}$.
	
	Consider first the case $n_0 = n_1$.  We have two candidates $\vec{s}_1$ and $\vec{s}_2$ for a maximizer of (\ref{eq:a_cost1}):$\vec{s}_1 = [0 \dotsc 01 \dotsc 1], \vec{s}_2 = [1 \dotsc 10 \dotsc 0].$ This is because a candidate has to start with either $0$ or $1$ and the remaining $k$ bits follow from the inductive assumption for sequences of length $k$. Both candidates $\vec{s}_1$ and $\vec{s}_2$ match (\ref{eq:a_maximizer}), i.e., the $\vec{s}_2$ symbols could be relabeled and achieve the same cost (\ref{eq:a_cost1}).
	
	Suppose next that $n_0 < n_1$. Again, we have only two\footnote{If $n_1 = n_0 +1$ we could have a third candidate $\vec{s}_3$ which would be the same as $\vec{s}_2$ but with all but the first bits inverted. However, the costs of $\vec{s}_2$ and $\vec{s}_3$ are the same due to uniform composition $\vec{\gamma}=[n_0,n_0]$ of the last $k$ symbols. Thus, we refrain from including $\vec{s}_3$ in the analysis.} candidates for a maximizer of (\ref{eq:a_cost1}):$\vec{s}_1 = [0 \dotsc 01 \dotsc 1], \vec{s}_2 = [10 \dotsc 01 \dotsc 1]$. 
	The candidate $\vec{s}_1$ has the desired form (\ref{eq:a_maximizer}) and we will show that it has a larger cost. Consider the cost difference $\Delta c = c(\vec{s}_1) - c(\vec{s}_2)$.
	We reduce the last $n_1-1$ terms of $c(\vec{s}_1)$ and $c(\vec{s}_2)$ that are the same. The difference $\Delta c$ becomes
	\begin{align}
	\Delta c &= \left( \sum_{i=1}^{n_0} f\left(\frac{x_i}{n_0+1-i}\right) -f\left(\frac{x_{i+1}}{n_0+1-i}\right) \right) + \nonumber\\
	&-\left(  f\left(\frac{x_1}{n_1}\right) - f\left(\frac{x_{n_0+1}}{n_1}\right)\right).
	\end{align}
	We can upper bound
	\begin{align}
	f\left(\frac{x_1}{n_1}\right) &- f\left(\frac{x_{n_0+1}}{n_1}\right) \stackrel{(\ref{eq:ap_f_property1})}{\le} f\left(\frac{x_1-x_{n_0+1}}{n_1}\right) \\
	&= f\left(\sum_{i=1}^{n_0}\frac{x_i-x_{i+1}}{n_1}\right) \stackrel{(\ref{eq:ap_f_property3})}{\le} \sum_{i=1}^{n_0} f\left(\frac{x_i-x_{i+1}}{n_1}\right).
	\end{align}
	This results in a lower bound 
	\begin{equation*}
	\Delta c \ge  \sum_{i=1}^{n_0} f\left(\frac{x_i}{n_0+1-i}\right) -f\left(\frac{x_{i+1}}{n_0+1-i}\right) - f\left(\frac{x_i-x_{i+1}}{n_1} \right). 
	\end{equation*}
	We consider the $i$-th summand for $f(x) = \log(1+\delta x)$. It follows that the summands are positive if for all $i=1,\dotsc,n_0$ we have
	\begin{equation*}
	\delta x_{i+1} < n_1 -n_0 -1 +i 
	\end{equation*}
	which is satisfied if $\delta x_{2} < 1$. This follows by the assumption $\frac{1}{x_1} \ge \delta > 0$. Thus, $\Delta c > 0$ and $\vec{s}_1$ is a unique maximizer of (\ref{eq:a_cost1}). \		
\end{proof}	
\fi

\begin{lemma}[Optimality of a greedy maximizer] \label{l:a_greedy_optmizer}
	Consider a sequence $\vec{a} \in \{a_1, \dotsc, a_m\}^n = \mathcal{A}^n$, a composition $\vec{\gamma} = [n_{a_1}, \dotsc, n_{a_m}]$ with $\sum_1^n n_{a_i}=n$, and a scalar function $f(x) = \log_2 \left( 1 + \delta x \right)$ with $\frac{1}{n} \ge \delta >0$. Define the cost function
	\begin{equation}
	c(\vec{s}) = \sum_{i=1}^n f\left( \frac{n+1-i}{n_{s_i} - n_{s_i}(\vec{s}_1^{i-1})}\right) 
	\end{equation}
	and consider the optimization
	\begin{equation}
	\max_{\vec{s} \in \mathcal{A}^n \colon \gamma(\vec{s}) = \vec{\gamma}} c(\vec{z}).
	\end{equation}
	Then  the greedy solution $\vec{z}$ defined below is a global maximizer of the optimization. The next symbol $z_i \in \mathcal{A}$ is obtained by
	\begin{equation}
	z_i = \argmin_{a \in \mathcal{A} \colon n_a > n_a(\vec{s}_1^{i-1})} n_a - n_a(\vec{s}_1^{i-1})
	\end{equation}
	where the constraint ensures that $\vec{z}$ has the required composition $\vec{\gamma}$ and the chosen $z_i$ maximizes the instantaneous cost increment.   	
\end{lemma}
\ifdefined\SHORT
\begin{proof}
	The proof was removed due to the $6$-page limit for submissions.
\end{proof}
\else
\begin{proof}
	The proof follows by contradiction. We first assume that there exists a maximizer $\vec{s}$ different then $\vec{z}$. Next, we show that $\vec{s}$ can be improved and hence the greedy solution $\vec{z}$ must be optimal.
	Consider the maximizer $\vec{s}$ and let $i_1$ be the first position where $\vec{s}$ and $\vec{z}$ differ. We define $y=s_{i_1}$, $x=z_{i_1}$, and $I_{xy} = \{i\colon s_i =x \lor s_i=y \} = \{i_1, \dotsc, i_k\}$. We construct a sequence $\vec{t}$ with $t_i = s_i$ for $i \notin I_{xy}$ and $(t_{i_1}, t_{i_2}, \dotsc, t_{i_k}) = (x, \dotsc, x, y, \dotsc, y)$, i.e, we have
	\begin{alignat}{6}
	\vec{s} &= [\dotsc &&\overbrace{s_{i_1}}^{=y} \dotsc &&s_{i_2}   \dotsc  &&s_{i_{k-1}} 	\dotsc &&s_{i_k} 	\dotsc &&]\\
	\vec{t} &= [\dotsc &&\;\;x 					 \dotsc  &&x 		 \dotsc  &&y 			\dotsc &&y 			\dotsc &&]. 
	\end{alignat} 
	Sequences $\vec{t}$ and $\vec{s}$ are the same on all positions except for $i \notin I_{xy}$, thus the difference $c(\vec{t})-c(\vec{s})$ depends only on $\vec{s}^\prime = [s_{i_1},\dotsc,s_{i_k}]$ and $\vec{t}^\prime = [t_{i_1},\dotsc,t_{i_k}]$. $\vec{s}^\prime$ and $\vec{t}^\prime$ have the same composition $\vec{\gamma}=[n_x^\prime, n_y^\prime]$ and $n_x^\prime < n_y^\prime$ because $\vec{z}$ is a greedy solution, and $i_1$ is the first position where $\vec{s}$ and $\vec{z}$ differ, i.e., $s_{i_1}=y, z_{i_1}=x$. From Lemma \ref{l:a_binary_miximizer} it follows that  $c(\vec{t})-c(\vec{s}) > 0$. Thus, we improved on the optimal solution $\vec{s}$, which proves the lemma. 
\end{proof}
\fi

Finally, we prove Theorem 1.
We are interested in the solution of 
\begin{equation}
\max_{\vec{c} \in {\mathcal{T}_{\vec{\gamma}}}} \;\; \sum_{i=1}^n \log_2\left(  1 + 2^{-w} \frac{n+1-i}{n_{c_i} - n_{c_i}(\vec{c}_1^{i-1})} \right).
\end{equation}
An FPA implementation of the CCDM requires $2^w \ge n$, since we require $2^w \ge \Theta$ (see (\ref{eq:2w_theta})) and for CCDM we use $\Theta = n-i$. This implies  $\frac{1}{n} \ge \delta=2^{-w} > 0$. From Lemma \ref{l:a_greedy_optmizer}, it follows that a greedy optimizer to the above problem is a global optimizer. Observe that the sequence (\ref{eq:ccdm_worst_case_seq}) is a greedy optimizer.   
\end{appendices}

\bibliographystyle{IEEEtran}
\vspace{-1mm}
\small{
\bibliography{C:/Users/marci/Desktop/CCDM_post_accept/comm}
}

\end{document}